\documentclass[pra,aps, twocolumn, preprintnumbers, amsmath,amssymb,secnumarabic, superscriptaddress]{revtex4}

\usepackage{graphicx, color, graphpap}

\usepackage{amsmath}
\usepackage{enumitem}
\usepackage{amssymb}
\usepackage{amsthm}
\usepackage{pstricks}
\usepackage{float}
\usepackage{multirow}
\usepackage{hyperref}
\usepackage[T1]{fontenc}

\usepackage{times}

\long\def\ca#1\cb{}

\newcommand{\ket}[1]{\vert #1\rangle}               
\newcommand{\bra}[1]{\langle #1\vert}              
\newcommand{\dya}[1]{\ket{#1}\!\bra{#1}}
\newcommand{\ip}[2]{\langle #1\vert#2\rangle}      

\newcommand{\EC}{\mathcal{E}}

\newcommand{\Tr}{{\rm tr}}

\renewcommand{\geq}{\geqslant}
\renewcommand{\leq}{\leqslant}
\newcommand{\mte}[2]{\langle#1|#2|#1\rangle }

\newcommand{\ot}{\otimes}

\def\id{{\leavevmode\rm 1\mkern -5.4mu I}}

\newcommand*{\Hmin}{H_{\min}}
\newcommand*{\Hmax}{H_{\max}}

\newcommand*{\hmax}{h_{\max}}

\newcommand{\al}{\alpha }

\newcommand{\Dl}{\Delta}

\renewcommand{\th}{\theta } 

\newcommand{\lm}{\lambda }

\newcommand{\sg}{\sigma }

\newtheoremstyle{example}{\topsep}{\topsep}%
{}
{}
{\bfseries}
{.}
{   }
{\thmname{#1}\thmnumber{ #2}}
\theoremstyle{example}

\theoremstyle{definition}
\newtheorem{theorem}{Theorem}

\newtheorem{lemma}[theorem]{Lemma}

\newcommand*{\tr}{\mathrm{tr}}

\newcommand*{\braket}[2]{\langle #1 | #2 \rangle}

\newcommand*{\cE}{\mathcal{E}}

\newcommand*{\cH}{\mathcal{H}}
\newcommand*{\cI}{\mathcal{I}}

\newcommand*{\sD}{\textsf{D}}
\newcommand*{\sE}{\textsf{E}}

\newcommand*{\se}{\textsf{e}}

\newcommand*{\etab}{\bar{\eta}}

\newcommand*{\bP}{\textsc{P}}
\newcommand*{\bQ}{\textsc{Q}}

\newcommand*{\bX}{\textsc{X}}
\newcommand*{\bY}{\textsc{Y}}
\newcommand*{\bZ}{\textsc{Z}}

\newcommand*{\RR}{R}

\newcommand{\rhoE}{\rho^{\cE}}

\usepackage{color}
\definecolor{myred}{rgb}{1,0,0}

\definecolor{myblue}{rgb}{0,0,0.8}

\definecolor{myyellow}{rgb}{0.9,0.8,0}

\definecolor{mygreen}{rgb}{0,0.6,0}

\definecolor{myorange}{rgb}{0.6,0.6,0}

\definecolor{mycerul}{rgb}{0,0.6,1}

\begin{document}

\title{State-dependent approach to entropic measurement-disturbance relations}

\author{Patrick J. Coles}
\email[]{pat@nus.edu.sg}
\affiliation{Centre for Quantum Technologies, National University of Singapore, 2 Science Drive 3, 117543 Singapore.}
\affiliation{Institute for Quantum Computing and Department of Physics and Astronomy, University of Waterloo, N2L 3G1 Waterloo, Ontario, Canada}

\author{Fabian \surname{Furrer}}
\email[]{furrer@eve.phys.s.u-tokyo.ac.jp}
\affiliation{Department of Physics, Graduate School of Science, University of Tokyo, 7-3-1
Hongo, Bunkyo-ku, Tokyo, Japan, 113-0033.}

\begin{abstract}
Heisenberg's intuition was that there should be a tradeoff between measuring a particle's position with greater
precision and disturbing its momentum. Recent formulations of this idea have focused on the question of
how well two complementary observables can be jointly measured. Here, we provide an alternative approach based on how enhancing the predictability of one observable necessarily disturbs a complementary one. Our measurement-disturbance relation refers to a clear operational scenario and is expressed by entropic quantities with clear statistical meaning. We show that our relation is perfectly tight for all measurement strengths in an existing experimental setup involving qubit measurements.
\end{abstract}

\maketitle

\section{Introduction}

Heisenberg's uncertainty principle~\cite{Heisenberg} is one of the most central concepts in quantum physics and with increasing experimental abilities to control quantum degrees of freedom it is no longer only interesting from a theoretical view; it is now practically relevant. For instance, it provides limits on quantum metrology \cite{Giovannetti:2011fk} and can be used to prove security in quantum cryptography~\cite{TLGR,PhysRevLett.109.100502}.
Moreover, experimental setups are now capable of sensitively testing such formulations \cite{ViennaNatPhys2012,TorontoPRL2012,Li11,Prevedel11, PhysRevLett.110.220402, PhysRevLett.112.020401}. These advances demand tight, operationally-meaningful formulations of the uncertainty principle.

The most common formulation of the uncertainty principle gives a limit on one's ability to prepare a system with low uncertainty for two complementary observables $X$ and $Z$. Textbooks often illustrate this with Robertson's \cite{Robertson} bound on the standard deviations
\begin{equation}\label{eq:Robertson}
\Dl X \Dl Z \geq \frac{1}{2} | \bra{\psi} [ X, Z ] \ket{\psi} | \, ,
\end{equation}
which generalised Kennard's \cite{kennard1927quantum} earlier relation for position and momentum observables $\Dl Q \Dl P \geq \hbar/2$.

A more subtle aspect of the uncertainty principle concerns not preparation limitations but rather measurement limitations \cite{Busch07}, for example, the idea that one cannot build a device that jointly measures $X$ and $Z$. Much progress has recently been made on quantitative tradeoffs for the accuracy of a joint measurement device~\cite{Busch13, BuschEtAl2014, PhysRevLett.112.050401,Renes2014}. One approach considers state-dependent errors using the root-mean-square (RMS) expectation of the noise operator \cite{Ozawa04, HallPhysRevA.69.052113, Branciard13} while a different approach considers calibrating the apparatus on idealised input states associated with the $X$ and $Z$ observables, using either RMS~\cite{Busch13, BuschEtAl2014} or information-theoretic \cite{PhysRevLett.112.050401, ColesEtAl} measures.

A different aspect of measurement uncertainty that will be the topic of this Letter considers a sequential measurement setting where one asks how measuring an observable $X$ disturbs the outcome of a future $Z$ measurement (see e.g.~\cite{Busch07}). On a more intuitive level this question was first discussed by Heisenberg~\cite{Heisenberg} for position and momentum observables. Modern approaches aim to rigorously prove a trade-off between extracting some amount of information about $X$ versus affecting a subsequent $Z$ measurement, a so-called measurement-disturbance relation (MDR). (We note that this should not be confused with a different approach~\cite{PhysRevA.53.2038, 0295-5075-77-4-40002} that considers the measurement-induced disturbance of the overall quantum state, rather than the disturbance of a specific observable $Z$.) Because sequential measurement can be thought of as an attempted joint measurement, a joint measurement relation also implies a corresponding MDR. MDRs associated with the noise-operator approach \cite{Ozawa04, HallPhysRevA.69.052113, Branciard13} have been the subject of interesting experimental tests~\cite{ViennaNatPhys2012,TorontoPRL2012, PhysRevLett.110.220402, PhysRevLett.112.020401}, and yet have also been recently criticised for lacking operational significance in the general case~\cite{KorzekwaEtAl2013, BuschEtAl2014,PhysRevA.89.022106, 2013arXiv1312.4393B}. This led to a major theoretical effort within the past year to find operationally meaningful MDRs \cite{Busch13, BuschEtAl2014, PhysRevLett.112.050401, Renes2014}. These very recent MDRs, which indeed have clear meanings, have abandoned the state-dependent nature of the noise-operator approach in favour of a state-independent approach.

On the other hand, we are interested in whether an operationally meaningful MDR exists for \textit{state-dependent} notions of error and disturbance. This has come under question since recent observations \cite{KorzekwaEtAl2013, 2013arXiv1312.4393B} imply that no non-trivial state-dependent MDR can be formulated for ``faithful'' error and disturbance measures, i.e., which only vanish if the measurement reproduces the statistics of a perfect $X$ measurement (in the case of error) and if there is no change to the statistical distribution for a subsequent $Z$ measurement (in the case of disturbance). We remark that here the typical notion of measurement error considers retrodiction, i.e., the accuracy of the $X$ measurement is judged by how well one can infer the \textit{past} value of $X$. However, Appleby showed for position and momentum observables~\cite{Appleby1,Appleby981} that there also exists a trade-off between disturbance and \textit{predictive} error.

In this Letter we present a state-dependent MDR involving a predictive measurement error for $X$ and a comparative disturbance measure for $Z$. The predictive error quantifies the correlation between the outcome of the measurement instrument and a future $X$ measurement. The disturbance compares the disturbed $Z$ distribution to the $Z$ distribution in the absence of the measurement instrument. Both of our errors, which are faithful and operational, are expressed by entropic quantities with significance in an information-theoretic context. As a novel extension we also consider the disturbance of the system's correlation with the environment. 

Furthermore, we apply our relation to position and momentum observbles $Q$ and $P$ and show a trade-off of the form
\begin{equation}
(Q\text{ precision}) \cdot (P\text{ disturbance}) \geq \frac{\hbar}{P \text{ initial uncertainty}}\nonumber \, 
\end{equation}
where the precision relates to the coarse graining of the position measurement. Thus, we rigorously capture Heisenberg's intuition, that measuring the position more precisely disturbs the momentum, and yet we emphasize that the tradeoff is weakened for input states with a high momentum uncertainty.

\section{Error and disturbance measures}

\begin{figure}[t] 
\begin{center}
\includegraphics[width = 8cm]{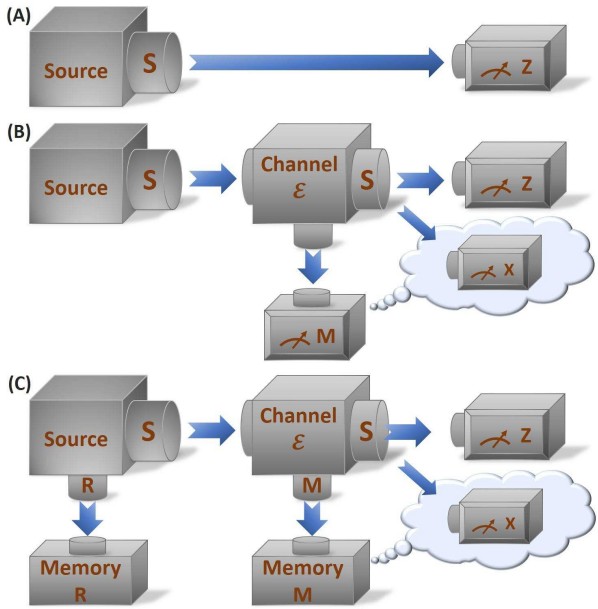}
\caption{(A): We consider a source $S$ sending a quantum state to a receiver which measures observable $\bZ$. (B): During the transmission a channel combined with a measurement is used to extract information $M$ about a second observable $\bX$. Our relation captures the tradeoff between the disturbance of the $Z$ distribution in (B) compared with the undisturbed situation (A) and the predictability of a hypothetical future measurement of $\bX$ given the information $M$. In the memory-assisted situation (C), $M$ can be a quantum memory and the disturbance of the correlations between $S$ and an isolated reference system $R$ is included. \label{fgrSetup1}}
\end{center}
\end{figure}

The physical scenario of interest is illustrated in Fig.~\ref{fgrSetup1}. We consider a system $S$ prepared in state $\rho_S$ and sent to a receiver who performs a measurement of the observable $\bZ$. During the transmission of $S$ to the receiver, an interaction $\cE$ is applied that intends to extract information about a complementary observable $\bX$. For simplicity, we assume for now that both observables $\bX$ and $\bZ$ are sharp and specified by orthonormal eigenstates $\{\ket{\bX_x}\}_{x\in X}$ and $\{\ket{\bZ_z}\}_{z\in Z}$, where $X$ and $Z$ are finite ranges. The treatment of more general observables is straightforward and considered later. The outputs of the interaction are the original system $S$ along with a classical system $M$ which is supposed to contain information that reduces the uncertainty about a future $\bX$ measurement. In the following we denote the $Z$ distribution of the initial state as $P_Z$ and the one after the interaction $\cE$ as $P^{\cE}_Z$. The joint probability distribution of $M$ and $X$ after the interaction is denoted by $Q^\cE_{MX}$.

Our goal is to define an operational measure for disturbance of the $Z$ degree of freedom which only depends on measurable quantities, that is, the probabilities $P_Z$ and $P_Z^\cE$. We further require that the disturbance is non-vanishing if and only if $P_Z \neq P_Z^\cE$ which we call a faithful disturbance measure. Natural candidates for faithful quantification of disturbance are thus distance measures between $P_Z$ and $P_Z^\cE$. A family of important information theoretic quantities with distance-like properties are the R\'enyi relative entropies~\cite{renyi1961} 
\begin{equation}
D_\alpha(P|| Q) = \frac{1}{\alpha-1}  \log \bigg( \sum_z   P(z)^\alpha Q(z)^{1-\alpha} \bigg) \, ,
\end{equation}
where $\alpha\in[1/2,\infty]$ and $\alpha=1,\infty$ are defined as the corresponding limit. Here and in the following all logarithms are in base 2. For $\alpha=1$ we obtain the well-known relative entropy (or Kullback-Leibler divergence) $D(P|| Q) = \sum_z P(z)\log ( P(z)/Q(z) )$. Another important example is obtained for $\alpha = 1/2$ where one finds that $D_{1/2}(P|| Q) = - \log F(Q,P)$ with $F(Q,P) = (\sum_z \sqrt{P(z)Q(z)})^2$ the fidelity between $P$ and $Q$. We note that the fidelity induces a metric equivalent to the statistical distance between $P$ and $Q$, and thus, the distinguishability of $P$ and $Q$. In what follows, we also need the R\'enyi entropies~\cite{renyi1961}, which can be defined through the R\'enyi relative entropies by 
\begin{equation}
H_\alpha(Z)_P = \log d - D_\alpha ( P|| \id / d ) \, ,
\end{equation} 
where $d$ is the Hilbert space dimension and $\id / d$ denotes the uniform distribution. For $\alpha=1$, we recover the Shannon entropy $H(Z)=-\sum_z p(z)\log p(z)$.

We now use the R\'enyi relative entropies to define a family of faithful disturbance measures 
\begin{align}\label{eq:DistMeas}
\sD_\alpha (\rho_S, \bZ, \cE) :=  D_\alpha( P_Z || P_Z^{\cE}) \, , 
\end{align}
where $\alpha\in[1/2,\infty]$. Particularly important in an information theoretic context is the relative entropy $D( P_Z || P_Z^{\cE})$ (i.e., $\alpha =1$), which quantifies to what extent large samples from $P_Z$ can be simulated by actually sampling from $ P_Z^{\cE}$. More precisely, it determines the probability to confuse sampling from $ P_Z^{\cE}$ with sampling from $P_Z$ for large samples, or the probability that a large sample from $ P_Z^{\cE}$ is typical for $P_Z$ (see e.g.~\cite{VedralReview02,Cover91}). This scenario has direct relevance in hypothesis testing~\cite{Cover91}. Moreover, a similar role is played by the R\'enyi relative entropies for $\alpha\neq 1$ if one considers higher order corrections or non-asymptotic behaviors~\cite{Csiszar1995}.

The performance of the interaction $\EC$ with respect to $X$ will be characterized by a predictive error, quantifying the degree of correlation between $M$ and a future $X$ measurement. Note that no non-trivial trade-off exists between our faithful disturbance measure and a retrodictive error, the usual measure of measurement accuracy (see also~\cite{KorzekwaEtAl2013}). One can see this by considering the following example: for any input state $\rho$, the interaction can simply consist of a perfect $X$ measurement followed by a (re)preparation of $\rho$, resulting in no apparent disturbance of $Z$. In contrast, a disturbance is obtained if we require that the interaction preserves the correlations between the extracted information $M$ and the $X$ eigenbasis, which directly translates into the predictability of a future $X$ measurement. This example also illustrates that an accurate $X$ measurement does not necessarily imply a small predictive error.

As in the case of the disturbance measure, the goal is to define an operational measure that only depends on the correlation of $M$ and the outocme of a future $X$ measurement determined by $Q^\cE_{MX}$. Moreover, the measure should be faithful in the sense that it is non-vanishing if and only if $M$ determines $X$ without uncertainty. From information theory, we know that conditional entropy measures are well-suited for that purpose. We thus quantify the predictive error by the conditional max-entropy~\cite{KonRenSch09}
\begin{align}\label{eq:ErrorMeas}
\sE(\rho_S, \bX, \cE) := H_{\max}(X|M)_{Q^\cE} \, ,
\end{align}
henceforth simply referred to as error. 

The conditional max-entropy is part of a family of entropies used to quantify resources beyond their behavior in the limit of infinitely many copies and is related to the amount of additional data that must be supplied to the observer, given that they have access to $M$, to learn the outcome of a future $\bX$ measurement \cite{RenesRenner2012}. In formulas the error is given by $ \log \sum_m Q^\cE_M(m) \exp (H_{1/2}(Q^{\cE,m}_{X} ))$, where $Q^\cE_M$ is the reduced probability distribution of $M$ and $Q^{\cE,m}_{X}$ is the conditional probability distribution of $X$ given $m\in M$.

\section{Measurement-disturbance relation}

Our main result gives a tradeoff between the $\bZ$ disturbance and the predictive error of the $\bX$ measurement. The tradeoff is stronger when $\bX$ and $\bZ$ are more complementary as quantified by the state-independent overlap
\begin{equation}\label{eq:Overlap}
c = \max_{x,z} |\braket{\bX_x}{\bZ_z}|^2 \, .
\end{equation}
Yet the trade-off is weaker as more initial uncertainty is contained in $P_Z$. More precisely, for any input state $\rho_S$ and interaction $\cE$, the MDR
\begin{align}
\label{eq1DiscURclas}
\sD_\alpha(\rho_S, \bZ, \cE)   + \sE(\rho_S, \bX, \cE) + H_{\alpha}(Z)_{P} \geq  \log 1/c  \, 
\end{align}
holds for all $\alpha\in [1/2,\infty]$.

We note that the additional term in \eqref{eq1DiscURclas} that quantifies the initial uncertainty of the $Z$ distribution is crucial. In order to see this consider for example the situation where $\bX$ and $\bZ$ are fully complementary, so-called mutually unbiased bases defined by $c = 1/d$, so that $\log 1/c = \log d$. Also suppose $\cE$ does a perfect $\bX$ measurement, so the error is zero. The disturbance is also zero, e.g., if $\rho_S$ is diagonal in any basis $\bY$ that is mutually unbiased to $\bZ$, since both the input and output probability distributions for $\bZ$ are uniform. We remark that, for this example, \eqref{eq1DiscURclas} is satisfied with equality \textit{for all input states} $\rho_S$. This is because doing an $\bX$ measurement followed by a $\bZ$ measurement always results in $\rhoE_Z = \id / d$, and we have $D_\alpha (\rho_Z || \id / d) = \log d - H_\alpha(Z)_{\rho}$.

For a given interaction $\cE$ and input state $\rho_S$, \eqref{eq1DiscURclas} actually represents two constraints: one given by \eqref{eq1DiscURclas} and another obtained from interchanging the roles of $\bX$ and $\bZ$ in \eqref{eq1DiscURclas}. For certain examples, one constraint may be significantly stronger than the other.

Before proving~\eqref{eq1DiscURclas}, we first discuss how our MDR can be further strengthened in two directions.
First, we can extend the scope of the MDR to include the disturbance of the system's correlations with a memory system, see Fig.~\ref{fgrSetup1}(C). Assume that system $S$ may be initially correlated to some other quantum system $\RR$, which we think of as an isolated memory system kept in the sender's lab while sending only $S$ to the receiver. The correlations between $S$ and $\RR$ may be disturbed by the interaction $\cE$. Let us denote the combined state of the quantum system $R$ and the classical outcomes $Z$ with and without interaction by $\rho^\cE_{ZR}$ and $\rho_{ZR}$. The memory-assisted disturbance is then defined as the distance between $\rho^\cE_{ZR}$ and $\rho_{ZR}$ defined for $\alpha\in [1/2,\infty]$ by 
\begin{equation} 
\sD_\alpha(\rho_{S\RR}, \bZ, \cE) = D_\alpha( \rho_{Z\RR} || \rho_{Z\RR}^{\cE}) \, ,
\end{equation}
where $D_\alpha(\rho||\sigma) = 1/(1-\alpha)\log \Tr [(\sg^{\frac{1-\al}{2\al}}\rho \sg^{\frac{1-\al}{2\al}})^\al]$ is a recently defined quantum generalization of the R\'enyi relative entropy~\cite{lennert13,Wilde13}. The cases $\alpha =1,\infty$ are defined as the respective limits.   

It is easy to see that setting $\RR$ to a trivial system recovers the previous (classical) notion of disturbance. For $\alpha=1$, we obtain the quantum relative entropy $D(\rho || \sg) = \Tr(\rho\log \rho) - \Tr(\rho\log\sg)$ and for $\alpha=1/2$ the logarithm of the quantum fidelity.  Moreover, the quantum relative entropy as well as the Renyi relative entropy retain their operational relevance to hypothesis testing, see e.g.,~\cite{VedralReview02,MosHia11,MosOga2013arXiv1309.3228M}. 

Replacing now the disturbance measure by its extended version $\sD_\alpha(\rho_{S\RR}, \bZ, \cE)$  in the MDR~\eqref{eq1DiscURclas}, the bound can be strengthened by replacing $H_\alpha(Z)$ on the left hand side by the corresponding conditional version~\cite{lennert13}
\begin{equation}\label{def:QMentropy}
H_\alpha(Z|R)_\rho = \max_{\eta_R}[- D_{\al}(\rho_{ZR} || \id \ot \eta_R)] \, . 
\end{equation}
This improves the bound since $H_\alpha(Z|R) \leq H_\alpha(Z)$ holds for any $\alpha $. Note that for $\alpha =1$, we obtain the conditional von Neumann entropy $H(A|B)_\rho = H(AB)_ \rho- H(B)_\rho$ with $H(A)\rho=-\tr\rho_A\log\rho_A$. Moreover, the quantum conditional min- and max-entropy which play an important role in non-asymptotic information theory (see e.g.~\cite{TomamichelThesis2012}) are given by $H_{\max}(A|B)_\rho =  H_{1/2}(A|B)_\rho$ and $H_{\min}(A|B)_\rho =  H_{\infty}(A|B)_\rho$. 

Secondly, we note that the system $M$ is not necessarily restricted to be classical but can be an arbitrary quantum system \cite{BertaEtAl}. In this case the distribution $Q^\cE_{XM}$ is replaced by a classical quantum state $\rho^{\cE}_{XM}$ and the predictive error is defined via the quantum conditional max-entropy $\sE(\rho_S, \bX, \cE) := H_{\max}(X|M)_{Q^\cE} $.  Since measuring the quantum system $M$ can only increase the uncertainty about $X$, this provides a stronger bound.

\begin{theorem} 
Let $S$, $M$ and $R$ be finite-dimensional quantum systems, $\rho_{SR}$ a quantum state on $SR$ and $\bX$ and $\bZ $ observables on $S$ given by positive operator valued measures $\{{\bX_x}\}_{x\in X}$ and $\{{\bZ_z}\}_{z\in Z}$, respectively. Then, for all trace preserving completely positive maps $\cE$ from $S$ to $SM$ and $\alpha\in[1/2,\infty]$, it holds that 
\begin{align}
\label{eq1DiscUR}
\sD_\alpha(\rho_{S\RR}, \bZ, \cE)   + \sE(\rho_{S}, \bX, \cE) +H_{\alpha}(Z|\RR)_{\rho} \geq  \log \frac1c   \,  , 
\end{align}
with $c= \max_{x,z} \|\sqrt{\bX_x}\sqrt{\bZ_z}\|_{\infty}^2$ and $\| \cdot \|_{\infty}$ the supremum norm (i.e., the largest singular value). 
\end{theorem}

\begin{proof}
An important technical ingredient in the derivation of the MDR is the inequality
\begin{equation}\label{MainInequ}
D_{\al}(\rho_{AB} || \sg_{AB})  \geq H_{\min}(A|B)_{\sg} - H_{\al}(A|B)_{\rho} \, ,
\end{equation}
which holds for any two states $\rho_{AB}$ and $\sg_{AB}$ and all $\alpha \in  [1/2,\infty ]$. This inequality follows from two basic properties of the Renyi relative entropies~\cite{lennert13}
\begin{align}
  D_{\al}(\rho||\sigma) &\geq  D_{\al}( \rho|| \eta) \text{  if  } \eta \geq \sigma , \label{Cond:a} \\
  D_{\al}(\rho|| \lambda \sigma)& = D_{\al}(\rho||\sigma)-\log \lambda \, \text{ for any  } \lambda>0\, ,  \label{Cond:b}
\end{align}
where the condition in \eqref{Cond:a} means that $\eta -\sigma$ is positive semi-definite, and from the fact that the min-entropy can be written as~\cite{KonRenSch09}
\begin{equation*}\label{eq:minEnt}
H_{\min}(A|B)_{\sg} =-\log  \min_{\eta_{B}} \min\{\lm : \sg_{AB} \leq \lm \id \otimes \eta_{B} \} .
\end{equation*}
In particular, let $\etab_B$ be a state for which the minimum in the above equation is attained, and thus, satisfies $\sg_{AB} \leq  2^{-H_{\min}(A|B)_{ \sg}} \id \ot \etab_B $. We can then compute from~\eqref{Cond:a} and~\eqref{Cond:b} that
\begin{align}
D_{\al}(\rho_{AB} || \sg_{AB}) &\geq D_{\al}(\rho_{AB} ||2^{-H_{\min}(A|B)_{ \sg} }\id \otimes \etab_B ) \notag \\
&  \geq D_{\al}(\rho_{AB} ||\id \otimes \etab_B )  + H_{\min}(A|B)_{ \sg } \, . \notag
\end{align}
Using the definition of $H_{\al}(A|B)_{\rho}$ from~\eqref{def:QMentropy}, we eventually find inequality~\eqref{MainInequ}.

The second ingredient is the preparation uncertainty relation with quantum memory for the min- and max-entropy~\cite{TomRen2010} applied to the state after the interaction $\rho^\cE_{SMR} = ( \cE \ot \cI)(\rho_{SR})$ 
\begin{equation} \label{eq:PrepUR}
\Hmin(Z|R)_{\rhoE} + \Hmax(X|M)_{\rhoE} \geq \log \frac{1}{c} \, .
\end{equation}
Combining the above relation with inequality~\eqref{MainInequ} for the case where $\rho_{AB}$ is replaced by $\rho_{ZR} $ and $\sg_{AB}$ by $\rhoE_{ZR} $, we arrive at our general memory-assisted MDR~\eqref{eq1DiscUR}. 
\end{proof}

\subsection{Predictions for experiments}

To demonstrate the performance of our MDR, we analyze a recent experiment on photon polarization~\cite{TorontoPRL2012} that tested Ozawa's MDR~\cite{Ozawa04}. In the following, we restrict ourselves to the disturbance measure for $\alpha =1$ based on the relative entropy, simply for illustration purposes.

The experiment implements a weak measurement of $\bX$, using a CNOT gate (controlled by the $\bX = \{\dya{+},\dya{-}\}$ basis) with the probe photon initially prepared in the state $\ket{\phi} = \cos (\th/2) \ket{0}+\sin (\th/2) \ket{1}$. The probe is then measured in the standard basis. Let us write the input state using the Bloch sphere representation $\rho_S = (\id +\vec{r} \cdot \vec{\sg})/2$, with $\vec{r} = r_x \hat{x}+r_y \hat{y}+r_z \hat{z}$. A straightforward computation gives 
\begin{align}
\sE(\rho_S, \bX, \cE) & = \log(1+\sqrt{1-r_x^2}\sin \th),  \label{eq:ErrorExp} \\ 
\sD (\rho_S, \bZ, \cE) & = \frac{1+r_z}{2} \log \big( \frac{1+ r_z}{1+r_z \sin \th}\big) \label{eq:DistExp}\\
&\hspace{10pt} +\frac{1-r_z}{2} \log \big( \frac{1- r_z}{1-r_z \sin \th} \big),\notag \\
H (Z)_{P}& =1- \frac{1+r_z}{2} \log \big( 1+ r_z \big) \label{eq:EntExp}\\
&\hspace{10pt} - \frac{1-r_z}{2} \log \big( 1- r_z \big). \notag
\end{align}

In Figure~\ref{fgr3Toronto}, we plot error, disturbance, and the tightness of our MDR \eqref{eq1DiscURclas} as a function of $r$ and $\th$, for an input state $\rho_S = (\id + r \sg_z)/2$. Notice that, when $\rho_S = \dya{0}$, corresponding to $r=1$, \eqref{eq1DiscURclas} is perfectly tight \textit{for all values of the measurement strength $\th$}, with $\sD(\rho_S, \bZ, \cE)   + \sE(\rho_S, \bX, \cE) =1$. Previous state-dependent MDRs \cite{Ozawa04,HallPhysRevA.69.052113,Branciard13} using r.h.s.\ of \eqref{eq:Robertson} for the bound are very untight in this case, giving a trivial bound. Note that a comparison to the literature is meaningful here since the interaction preserves the $\bX$ observable hence predictive and retrodictive error are identical.

\begin{figure}[t] 
\begin{center}
\includegraphics{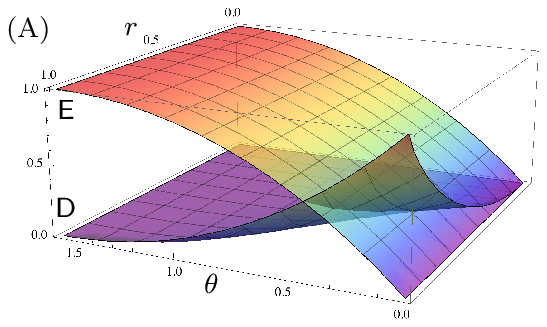}
\includegraphics{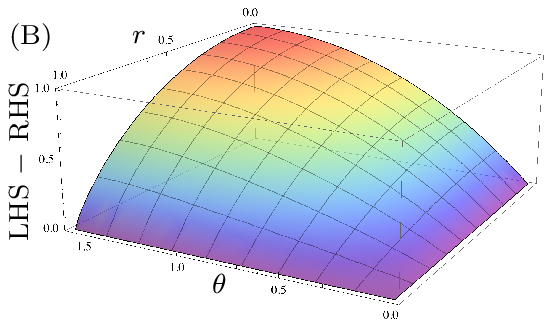}
\caption{For the photon polarization experiment in \cite{TorontoPRL2012}, we plot (A) the error $\sE(\rho_S, \bX, \cE)$ and disturbance $\sD(\rho_S, \bZ, \cE) $, and (B) the difference between the left- and right-hand-sides of \eqref{eq1DiscURclas}. We consider input states $\rho_S = (\id + r \sg_z)/2$ and $\th$ determines the measurement strength (see text; $\th =0$: perfect $\bX$ measurement, $\th = \pi/2$: completely noisy $\bX$ measurement).\label{fgr3Toronto}}
\end{center}
\end{figure}

Let us now consider the memory-assisted relation~\eqref{eq1DiscUR}. It turns out that the experimenters chose the optimal measurement (the standard basis) on $M$ for minimising our error measure. To show this let us define the state $\ket{\chi}=\ket{\chi_u}/\sqrt{N}$ with $\ket{\chi_u} = \sum_x \sqrt{p_{x}} \ket{\phi_x}$, $N = \ip{\chi_u}{\chi_u}$ and $p_{x} = \mte{\bX_x}{\rho_S}$. Chosing now $\ket{\chi}$ as a candidate to achieve the maximization in the definition of the quantum conditional max-entropy (i.e.~\eqref{def:QMentropy} for $\alpha=1/2$), we obtain a lower bound on the error which turns out to be equal to~\eqref{eq:ErrorExp}. Since measuring system $M$ has to increase $H_{\max}(X|M)$, we can conclude that $H_{\max}(X|M)$ actually coincides with~\eqref{eq:ErrorExp}. 

However, the situation changes dramatically if we allow the sender to possess a memory photon, whose polarisation $\RR$ is initially correlated to polarisation $S$. For example, suppose $\rho_S = (\id + r \sg_z)/2 $ as in Fig.~\ref{fgr3Toronto} and let $\rho_{S\RR}$ be such that $\RR$ is perfectly correlated to the $\bZ$ observable on $S$. Then, we have $\sD(\rho_{S\RR}, \bZ, \cE) = 1-\log (1+\sin \th)$, $\sE(\rho_S, \bX, \cE) = \log (1+\sin \th) $, and $H(Z|\RR )_{\rho} = 0$. In other words, \eqref{eq1DiscUR} is satisfied with equality for all values of the parameters $r$ and $\th$ with $\sD(\rho_{S\RR}, \bZ, \cE) + \sE(\rho_S, \bX, \cE) =1$. We note that, e.g., setting $\RR$ to be a classical system would allow experimentalists to easily test our memory-assisted MDR, and the fact that \eqref{eq1DiscUR} is actually an equation for all $r$ and $\th$ would make it a highly sensitive test. 

\section{Position and momentum}

Let us now consider our MDR for Heisenberg's original setup of position and momentum observables $\bQ$ and $\bP$. We consider two situations, namely, coarse grained and continuous outcome measurements. In the former situation, we have positive operator valued measures corresponding to position and momentum projections onto intervals of length $\delta q$ and $\delta p$, which we denote by $\bQ_{\delta q}$ and $\bP_{\delta p}$, respectively. The outcome range of these measurements $Q_{\delta q}$ and $P_{\delta p}$ are discrete but infinite, and each outcome relates to a position and, respectively, momentum in a unique interval of length $\delta q$ and $\delta p$. The situation of continuous outcomes which we refer simply by $Q$ and $P$ can then be seen as obtained in the limit $\delta q, \delta p \rightarrow 0$. 

First, we generalize the definition of disturbance and predictive error to the infinite-dimensional setting. Due to technical reasons, we focus on the disturbance measure based on the quantum relative entropy ($\alpha=1$). The quantum relative entropy is well studied in infinite-dimensional systems and we use the definition based on the spatial derivative operator (see, e.g.,~\cite{Petz93,Berta13} and references therein). This definition applies similarly for continuous and discretized outcome measurements such that $\sD(\rho_{SR}, \bP, \cE)$ and  $\sD(\rho_{SR}, \bP_{\delta p}, \cE)$ can be defined similarly as in the finite-dimensional case via~\eqref{eq:DistMeas}. Moreover, as shown in Lemma~\ref{lem:ApprxRel}, it holds that   
\begin{equation}\label{eq:approxRelEnt}
\lim_{\delta p \rightarrow 0} \sD(\rho_{SR}, \bP_{\delta p}, \cE) = \sD(\rho_{SR}, \bP, \cE)
\end{equation} 
for any $\cE$. 

Using the general definition of the relative entropy, the von Neumann entropy of $\bP_{\delta p}$ conditioned on the quantum system $R$ can again be defined as in finite dimensions via the relation~\eqref{def:QMentropy}. The differential version of the conditional von Neumann entropy is defined as in~\cite{Berta13} and essentially given by the regularised limit
\begin{equation}
h_{}(P|R)_{\rho} = \lim_{\delta p \rightarrow 0 } \big( H(P_{\delta p}|R) + \log \delta p \big) \,   . 
\end{equation} 

In order to define the predictive error, we use the quantum conditional max-entropy $H_{\max}(Q_{\delta q}|M)_\rho$ as introduced in~\cite{Berta13}. The differential version is defined as in~\cite{Berta13} and under weak assumption is simply given by the limit  $h_{\max}(Q|M)_\rho = \lim_{\delta q \rightarrow 0} [H_{\max}(Q_{\delta q}|M)_\rho + \log \delta q] $. Emphasizing that the involved error is now a regularised quantity we denote it with a lower case letter
\begin{align}\label{eq:ErrorMeasPQ}
\se (\rho_S, \bQ, \cE) &:= h_{\max}(Q|M)_{\rhoE} \, .
\end{align}

\begin{theorem}
Let $S$ be a position-momentum system and $R$ and $M$ quantum systems described by separable Hilbert spaces $\cH_R$ and $\cH_M$. For any state $\rho_{SMR}$ on systems $SMR$ and any completely positive map $\cE$ from $S$ to $SM$, it holds that 
\begin{align}\label{eq:DiscMDRpq}
\sD(\rho_{SR}, \bP_{\delta p}, \cE)   + \sE(\rho_{S}, \bQ_{\delta q}, \cE)  & + H(P_{\delta p}|R)_{\rho} \\  \quad &\geq  \log c(\delta q,\delta p) \, , \nonumber 
\end{align}
where $c(\delta q,\delta p) = \delta q \delta p /(2\pi\hbar) S_0^{(1)}(1,\delta q \delta p/4) \approx (\delta q \delta p)/(2\pi\hbar)$ with $S_0^{(1)}(1,\cdot)$ the 0th radial prolate spheroidal wave function of the first kind~\cite{Slepian1964}.
Moreover, if $\delta,\delta'$ exists such that $\sD(\rho_{SR}, \bP_{\delta}, \cE) < \infty$ and $H(P_{\delta'}|R)_{\rho}<\infty$, and $h_{}(P|R)_{\rho} \geq -\infty$, it holds that
\begin{align}
\label{eq:CVMDR55}
\sD(\rho_{SR}, \bP, \cE)   + \se(\rho_{S}, \bQ, \cE) +h_{}(P|R)_{\rho}  \geq   \log 2\pi\hbar   \, .
\end{align}
\end{theorem}
\begin{proof}
The proof for coarse grained position measurements is in complete analogy to the one in the finite-dimensional setting.  In particular, inequality~\eqref{eq:PrepUR} was shown in~\cite{Berta13} with $c$ given by $c(\delta q,\delta p)$. Moreover, inequality~\eqref{MainInequ} is valid for $\alpha =1 $ since properties~\eqref{Cond:a} and~\eqref{Cond:b} remain true in this more general setting~\cite{Petz93}. The only subtle point is that in~\cite{Petz93,Berta13} the conditional von Neumann entropy is defined as $H(X|B) = -D(\rho_{XB}|| \id_X \otimes \rho_B)$ and not via an optimization as in~\eqref{def:QMentropy}, which was used to arrive at~\eqref{MainInequ}. But the equivalence of these two definitions is shown in Lemma~\ref{lem:InfCondvNentropy} via a chain rule of the relative entropy. Finally, the version for continuous outcome measurements is obtained by taking the limit $\delta q,\delta p\rightarrow \infty$ using~\eqref{eq:approxRelEnt} and results from~\cite{Berta13}. 
\end{proof}

In the following, we illustrate our position-momentum MDR with two examples.

\subsection{The Heisenberg microscope}

Consider the situation in which the interaction $\cE$ corresponds to an instrument that performs a projective coarse grained measurement with discretization $\delta q$, i.e., a coherent measurement of $\bQ_{\delta q}$. Applying our MDR for $\bQ_{\delta q}$ and $\bP_{\delta p}$ in~\eqref{eq:DiscMDRpq} and taking the limit $\delta p\rightarrow 0$, we find that the error term $\sE(\rho_S, \bQ_{\delta q}, \cE)$ vanishes since the measurement is repeatable, and the complementarity constant is given by $c = \delta q /(2\pi\hbar)$. Hence, assuming that the $R$ system is trivial and taking the logarithm of our MDR we arrive at a relation of the form
\begin{align} \label{eq:ProdEDR}
\delta q \cdot d_p  \geq   \hbar / 2  \, .
\end{align}
This looks similar to~\eqref{eq:Robertson} where the error is given by $\delta q$ and the disturbance by $d_p := {2^{ h(P)_{\rho}}} 2^{\sD(\rho_{S}, \bP, \cE)}/(4 \pi)$. However,
$d_p$ is lower bounded by $2^{ h(P)_{\rho}}/(4 \pi)$, and thus, accounts for the uncertainty of the initial momentum distribution. Note that this is necessary since we can always choose an initial wave function that is confined to one measurement bin, i.e., with a position standard deviation much smaller than $\delta q$. Thus, no momentum disturbance results from the measurement. But this comes at the cost of a high initial momentum uncertainty, revealing an interesting interplay between preparation and measurement uncertainty.

\subsection{Covariant approximate position measurements}

As a second example, we consider an experimental setup that can be implemented using for instance optical systems~\cite{GrangierNat98,Porta89}. The interaction is given by a quantum nondemolition measurement implementing a covariant approximate position measurement discussed by von Neumann~\cite{vonNeumann} and Davies~\cite{Davies}. In particular, we assume that $S$ interacts with a similar meter system $M$ through a Gaussian operation acting in the Heisenberg picture according to $(\hat Q,\hat P,\hat Q',\hat P') \mapsto (\hat Q,\hat P-\hat P',\hat Q' + \hat Q,\hat P')$, where $\hat Q,\hat P$ and $\hat Q' , \hat P'$ denote position and momentum operators of system $S$ and $M$, respectively. After the interaction, the position of the meter system is measured. 

If the input state on $S$ and $M$ are assumed to be pure Gaussian states with position variance $V_S$ and $V_M$, respectively, the disturbance and error can be explicitly calculated. In the following, the parameter $\lambda := V_M/V_S$ can be interpreted as the effective resolution of the approximate position measurement~\cite{Davies}. 
The error term is given by $\se(\rho_{S}, \bQ, \cE) =  h_{\max}(Q|Q')_{\rho} $, where $\rho_{QQ'}$ denotes the joint probability distribution of $Q$ and $Q'$ after the interaction. Denoting the wave function of the initial state for $S$ and $M$ by $\psi_S$ and $\xi_M$, it is straightforward to see that $\rho_{QQ'}(q,q')= \vert \psi_S(q)\xi_M(q'-q)\vert^2$. Using the formula for the differential conditional max-entropy in Lemma~\ref{lem:ClassCondMax}, a simple computation gives $\se(\rho_{S}, \bQ, \cE) = \log 2 \sqrt{{ 2\pi V_S }/({1 + 1/\lambda)}} $. 

The disturbance $\sD(\rho_{S}, \bP, \cE) = D(\rho_P||\rho_P^\cE)$ can be computed by noting that $\rho_P(p)= |\hat \psi_S(p)|^2$ and $\rho_{PP'}^\cE =|\hat\psi(p+p') \hat\xi_M(p')|^2$, where $\hat f$ denotes the Fourier transform of $f$.  We then find that $D(\rho_P||\rho_P^\cE) =   - h(P)_\rho + \log[{\hbar}/{2}\sqrt{{2\pi(1+1/\lambda)}/{V_S}}] 
 +  1/[{2 \ln (2) (1+1/\lambda)}]$. 

We can now analyze the tightness of the MDR in~\eqref{eq:CVMDR55}. Computing the gap between the l.h.s. and r.h.s of~\eqref{eq:CVMDR55} gives $1 /(2 \ln 2) (1+1/\lambda)^{-1}$. Thus, the gap depends only on the effective resolution $\lambda$ and closes as $\lambda$ approaches $0$ proving tightness of our MDR~\eqref{eq:CVMDR55}.

\section{Conclusion}
We presented a state-dependent measurement disturbance relation, which in contrast to most previous relations, includes a predictive error rather than measurement accuracy. The disturbance as well as the predictive error are quantified by entropic quantities with clear statistical meaning. We demonstrated the tightness of our MDR with various examples including approximate position and momentum observables. We further introduced the novel concept of memory-assisted disturbance, where a quantum memory helps to reveal the disturbing effects of a measurement; this idea could be further explored using other measures or approaches. 

We remark that the factor $\log (1/c)$ in \eqref{eq1DiscUR} might be improved upon when $\bX$ and $\bZ$ are not MUBs. For example, for $\al = 1$ it can be replaced by a stronger bound using the approach of \cite{coles14}. Majorization approaches \cite{partovi11, puchala13, friedland13} might also be useful along these lines.

We further remark that our approach is closely connected to the security of reverse reconciliation quantum key distribution protocols in which the key is extracted from the receiver's measurement data. Such protocols are essential for long distance continuous variable quantum key distribution~\cite{Grosshans03}, due to their robustness against fiber losses. In our setup, the eavesdropper is modeled by the channel $\cE$ that extracts information $M$ about the receiver's $X$ measurement. In order to detect the leaked information $M$ the receiver applies randomly a test measurement $Z$. The security is then obtained by lower bounding the uncertainty of $X$ given the eavesdropper's information $M$ by means of the channel disturbance. But this is exactly the trade-off characterized by our measurement-disturbance relation. This close relation of our approach to measurement disturbance and quantum key distribution may lead to future applications in quantum cryptography.

\appendix

\section{Technical Lemmas}

For the following approximation result, we define a coarse graining of $X=\mathbb R$ as a family of finer and finer partitions of $X$ into disjoint intervals of length $\delta = 1/2^n$ parametrised by $n\in \mathbb N$. The intervals are further defined recursively by halving every interval in the step $n$ to $n+1$. For a more detailed discussion we refer to~\cite{Berta13}.

\begin{lemma}\label{lem:ApprxRel}
Let $\rho_{XB}$ and $\sigma_{XB}$ be continuous classical quantum states over $X=\mathbb R$ with $\cH_B$ a separable Hilbert space. If $ D(\rho_{X B}||\sigma_{X B}) $ is finite, then it holds that
\begin{equation}\label{lem,eq1:ApprxMaxRel}
\lim_{\delta \rightarrow 0} D(\rho_{X_\delta B}||\sigma_{X_\delta B}) = D(\rho_{X B}||\sigma_{X B})  \, ,
\end{equation}
where the limit is taken along a coarse graining of $X$.
\end{lemma}

\begin{proof}
Let us fix an arbitrary partition in a coarse graining of $X$ with intervals of length $\delta_0$ and denote the intervals by $X^k$. By using the disintegration theory for von Neumann algebras~\cite{Takesaki1}, we get by the monotone convergence theorem that
\begin{align}
D(\rho_{X B}||\sigma_{X B})  & = \int D(\rho_B^x || \sigma^x_B) d x \\
&  = \sum_k \int_{X^k} D(\rho_B^x || \sigma^x_B) d x \\
&  = \sum_k D(\rho_{X^k B}||\sigma_{X^k B})  \, , \label{eq1,pf:ApprxRel}
\end{align}
where $\rho_{X^kB}$ denotes the state projected to the interval $X^k$ and likewise for $\sigma_{X^kB}$. Since $X^k$ is compact for every $k$, we can use the approximation result from~\cite[Corollary 5.12]{Petz93} along an increasing net of subalgebras which generates $L^\infty(X^k)$ in the $\sigma$-weak topology. Such a net of subalgebras is given by the step-function over the partitions in the coarse graining with $\delta\leq \delta_0$ which implies that
\begin{equation}
\lim_{\delta \rightarrow 0} D(\rho_{X^k_\delta B}||\sigma_{X^k_\delta B}) = D(\rho_{X B}||\sigma_{X B})  \, ,
\end{equation}
for every $k$. 

Hence it remains to exchange the infinite sum in~\eqref{eq1,pf:ApprxRel} with the limit $\delta\rightarrow 0$. For that, we verify Weierstasse' uniform convergence criterion for infinite sums by finding a uniform upper bound on $g_k(\delta) = | D(\rho_{X^k_\delta B}||\sigma_{X^k_\delta B}) | \leq M_k$ such that $\sum M_k < \infty$. By the monotonicity of the quantum relative entropy under quantum channels (see, e.g.,~\cite[Corollary 5.12]{Petz93}), we obtain
\begin{equation*}
D(\tr(\rho_{X^k_\delta B}) || \tr(\sigma_{X^k_\delta B}) ) \leq g_k(\delta) \leq D(\rho_{X^k B}||\sigma_{X^k B}) \, .
\end{equation*}
Denoting $p_k = \tr(\rho_{X^k_\delta B})$ and $q_K=\tr(\sigma_{X^k_\delta B})$, we have that the left hand side is given by $p_k\log(p_k/q_k)$ which is only strictly smaller than $0$ if $q_k > p_k$. Hence, if we define $M_k = D(\rho_{X^k B}||\sigma_{X^k B})$ if $p_k \geq q_k$ and $M_k = D(\rho_{X^k B}||\sigma_{X^k B}) +  p_k\log(q_k/p_k)$ else, we obtain that $g_k(\delta)\leq M_k$ for all $k$.  We then find that
\begin{align}
\sum_k D(\rho_{X^k B}||\sigma_{X^k B})  = D(\rho_{X B}||\sigma_{X B})  < \infty
\end{align}
by assumption. 

Moreover, if we denote $\Gamma = \{k \ | \ q_k > p_k\}$ we get that
\begin{align}
\sum_{k\in \Gamma} p_k\log(q_k/p_k) &\leq \frac{1}{\ln 2} \sum_{k\in \Gamma} p_k( \frac {q_k}{p_k} - 1) \\
&\leq \frac{1}{\ln 2}\sum_{k\in \Gamma} q_k  \, ,
\end{align}
where we used the bound $\log x \leq \frac{1}{\ln 2}(x-1)$. Hence, we find that $\sum M_k <\infty$ which completes the proof.
\end{proof}

\begin{lemma}\label{lem:InfCondvNentropy}
Let $\rho_{XB} = \sum_{x}\ket x \bra x \otimes \rho_B^x$ be a normalised classical quantum state where the classical system $X$ is discrete but possibly infinite and $\cH_B$ separable. We then have that
\begin{equation}
H(X|B) = - \inf_{\sigma_B}  D(\rho_{XB} || \id_X \otimes \rho_B) \, ,
\end{equation}
where $H(X|B) =  - \inf_{\sigma_B} D(\rho_{XB} || \id_X\otimes \sigma_B)$ and the maximization is taken over normalised density matrices $\sigma_B$.
\end{lemma}
\begin{proof}
The claim is a direct consequence of the chain rule
\begin{equation}\label{eq:chainrule2}
D(\rho_{XB} || \eta_X \otimes \sigma_B) = D(\rho_{XB} || \eta_X \otimes \rho_B) + D(\rho_B||\sigma_B) \, ,
\end{equation}
which has been proven for non-normalised density operators~\cite[Corollary 5.20]{Petz93}. Note that if $X$ has infinite cardinality $\id_X$ is no longer a density matrix. In order to circumvent this problem we now use a limit argument. 

In the following we assume that $X=\mathbb N$. Let us define $X_n=\{1,2,...,n\} \subset X$ and $\rho^n_{X_n B}= \sum_{x\leq n} \ket x \bra x \otimes \rho_B^x$ the non-normalised state given by restricting onto $X_n$.  For every $n$ we can now apply the chain rule~\eqref{eq:chainrule2}
\begin{align*}
& D(\rho^n_{X_nB} || \id_{X_n} \otimes \sigma_B) \\
& = D(\rho^n_{X_nB} || \id_{X_n} \otimes \rho^n_B) + D(\rho^n_B||\sigma_B) \\
& \geq  D(\rho^n_{X_nB} || \id_{X_n} \otimes \rho_B) + D(\rho^n_B||\sigma_B)  \, ,
\end{align*}
where the inequality is obtained since $\rho_B \geq \rho^n_B$ and the monotonicity of the relative entropy $D(\rho||\sigma)\geq D(\rho||\eta)$ if $\eta \geq \sigma$ (see, e.g.,~\cite{Petz93}). Note now that for any $\sigma_B$ holds that $D(\rho^n_{X_nB} || \id_{X_n} \otimes \sigma_B) = \sum_{x=1}^n D(\rho_B^x||\sigma_B)$ and $D(\rho_{XB} || \id_{X} \otimes \sigma_B) = \sum_{x=1}^\infty D(\rho_B^x||\sigma_B)$, where all the terms $D(\rho_B^x||\sigma_B)$ are negative. We therefore have that $D(\rho^n_{X_nB} || \id_{X_n} \otimes \sigma_B) \rightarrow D(\rho_{XB} || \id_{X} \otimes \sigma_B)$ for $n\rightarrow \infty$~\cite{Berta13}. Taking now the limit inferior on both sides of the equation, we get that
\begin{align*}
D(\rho_{XB} || \id_{X} \otimes \sigma_B) & \geq D(\rho_{XB} || \id_{X} \otimes \rho_B) + D(\rho_B||\sigma_B)   \, ,
\end{align*}
where we used that the quantum relative entropy is lower semi-continuous, that is, $\liminf_{n\rightarrow \infty} D(\rho^n_B||\sigma_B)  \geq D(\rho_B||\sigma_B)$~\cite[Corollary 5.12]{Petz93}. Since $D(\rho_B||\sigma_B)\geq 0$ with equality if and only if $\sigma_B=\rho_B$, this establishes $\inf_{\sigma_B}  D(\rho_{XB} || \id_X\otimes \sigma_B) \geq D(\rho_{XB} || \id_{X} \otimes \rho_B) $, and thus, the claim.
\end{proof}

For the following we note that the probability distributions over $\mathbb R$ are given by the positive, normalized and measurable functions. We denote the Banach space of measurable functions $f$ such that $\int |f(x)|^p dx $ is finite by $L^p(\mathbb R)$, where $0<p<\infty$. The next Lemma generalizes a result in~\cite{TomamichelThesis2012} for finite and discrete $X$ and $Y$ to $\mathbb R$.

\begin{lemma}\label{lem:ClassCondMax}
Let $X = Y = \mathbb R$ and $P \in L^1(X\times Y)$ be a joint probability distribution such that $\hmax(X)_{P} < \infty$ and $ \hmax(Y)_P <\infty$. Then, it holds that
\begin{equation}\label{lem,eq1:ClassCondMax}
\hmax(X|Y)_P = \log \int dy \Big( \int dx \sqrt{P(x,y)} \Big)^2
\end{equation}
if the integral on the right hand side is finite.
\end{lemma}
\begin{proof}
By the definition of the differential conditional max-entropy~\cite{Berta13}, we have that
\begin{equation}\label{lem,pf1:ClassCondMax}
\hmax(X|Y) = 2 \log \sup_{q} \int dx \int dy \sqrt{P(x,y)q(y)} \, ,
\end{equation}
where the supremum is taken over all probability distributions $q\in L^1(Y)$. Note that $\hmax(X|Y) \leq \hmax(X)_{P} < \infty$ and $\hmax(Y|X) \leq \hmax(Y)_P <\infty$ implies that for any probability distribution $q$ the integrals $\int (\int \sqrt{P(x,y)q(y)} dy)dx $ and $\int (\int \sqrt{P(x,y)q(y)} dx)dy $ are finite. Hence, by Fubini's theorem we can interchange the integrations to get
\begin{align*}
 & \sup_{q} \int dx \int dy \sqrt{P(x,y)q(y)}  \\
& =   \sup_{q} \int dy \Big(\int dx \sqrt{P(x,y)} \Big) \sqrt{q(y)} \, .
\end{align*}
Let us define $\phi(y) = \int dx \sqrt{P(x,y)}$ which is in $L^2(Y)$ by assumption. Since $\Vert \sqrt{q} \Vert_{L^2(Y)} = \Vert q \Vert^{1/2}_{L^1(Y)} = 1$ we have that $\sqrt{q}\in L^2(Y)$. Using that $L^2(Y)$ is a Hilbert space, we get that
\begin{align*}
\sup_{q} \int dy \phi(y) \sqrt{q(y)}  & =  \sup_{q}  \braket { \phi } {\sqrt{q}} \\
& \leq  \sup_{q}  \Vert \phi \Vert_{L^2(Y)} \Vert \sqrt{q} \Vert_{L^2(Y)} \\
& \leq \Vert \phi \Vert_{L^2(Y)} \, .
\end{align*}
Further if we take for $q$ the element $q^*$ defined via $\sqrt{q^*} = \phi/\Vert \phi \Vert_{L^2(Y)}$  the maximum is attained. Plugging in $q^*$ in~\eqref{lem,pf1:ClassCondMax} we obtain~\eqref{lem,eq1:ClassCondMax}.
\end{proof}

\section*{Acknowledgments}
We thank M. Hall, F. Buscemi, and M. Wilde for helpful comments on an earlier version. PJC is funded by the Ministry of Education (MOE) and National Research Foundation Singapore, as well as MOE Tier 3 Grant ``Random numbers from quantum processes" (MOE2012-T3-1-009).  FF acknowledges support from Japan Society for the Promotion of Science (JSPS) by KAKENHI grant No. 24-02793.

\bibliographystyle{elsarticle-num}

\end{document}